\documentclass[conference,a4paper]{IEEEtran}   
\usepackage[english]{babel}
\usepackage{amssymb}
\usepackage{amsmath}
\usepackage{amsthm}
\usepackage{epsfig}
\usepackage{graphicx}
\usepackage{color}
\usepackage{dsfont}
\usepackage{cleveref}
\usepackage{subcaption}
\usepackage{enumerate}
\usepackage[backend=bibtex,giveninits=true,sorting=none]{biblatex}
\usepackage{multicol}
\usepackage{caption}
\usepackage{balance}
\renewcommand*{\bibfont}{\footnotesize}

\long\def\comment#1{}

\newfont{\bbb}{msbm10 scaled 700}

\newfont{\bb}{msbm10 scaled 1100}

\newcommand{\PP}{\mathbb{P}}
\newcommand{\RR}{\mathbb{R}}

\newcommand{\av}{{\bf a}}

\newcommand{\mv}{{\bf m}}

\newcommand{\sv}{{\bf s}}

\newcommand{\xv}{{\bf x}}
\newcommand{\yv}{{\bf y}}
\newcommand{\zv}{{\bf z}}

\newcommand{\Am}{{\bf A}}

\newcommand{\rhov}{\hbox{\boldmath$\rho$}}

\newcommand{\SNR}{{\sf SNR}}

\DeclareMathOperator*{\argmax}{arg\,max}

\newtheorem{theorem}{Theorem}

\newtheorem{remark}{Remark}
\begin{document}

\newcommand\figref{Figure~\ref}

\newcommand{\ben}{\begin{enumerate}}
\newcommand{\een}{\end{enumerate}}

\newcommand{\beq}{\begin{equation}}
\newcommand{\eeq}{\end{equation}}

\newcommand{\bi}{\begin{itemize}}
\newcommand{\ei}{\end{itemize}}

\newcommand{\1}{\mathds{1}}
\markright{\today}

    \title{Unsourced Multiuser Sparse Regression Codes achieve the Symmetric MAC Capacity
        }
\author{\IEEEauthorblockN{Alexander Fengler, Peter Jung, and Giuseppe Caire}
\IEEEauthorblockA{Communications and Information Theory Group,
Technische Universit\"{a}t Berlin\\
E-mail: \{fengler, peter.jung, caire\}@tu-berlin.de.}
}
\maketitle

\begin{abstract}
    Unsourced random-access (U-RA) is a type of grant-free random access with a
    virtually unlimited number of users, of which only a certain number $K_a$ are
    active on the same time slot. Users employ exactly the same codebook, and
    the task of the receiver is to decode the list of transmitted messages.
    Recently a concatenated coding
    construction for U-RA on the AWGN channel was presented,
    in which a sparse regression code (SPARC) is used as an
    inner code to create an effective outer OR-channel. Then an outer code is used
    to resolve the multiple-access interference in the OR-MAC. In this work
    we show that this concatenated construction can achieve a vanishing per-user error probability
    in the limit of large blocklength and a large number of active users at sum-rates up to the symmetric
    Shannon capacity, i.e. as long as $K_aR < 0.5\log_2(1+K_a\SNR)$.
    This extends previous point-to-point optimality results about SPARCs to the unsourced
    multiuser scenario.
    Additionally, we calculate the algorithmic threshold, that is a bound on the sum-rate
    up to which the inner decoding can be done reliably with the low-complexity AMP algorithm.
\end{abstract}

\section{Introduction}
One of the key new application scenarios of future wireless networks is
known as the Internet-of-Things (IoT), where it is envisioned that
a very large number of devices (referred to as users) is sending data to a common
access point. 
Typical examples thereof include sensors for monitoring smart infrastructure or biomedical devices. 
This type of communication is characterized by short messages and sporadic activity.
The large number of users and the sporadic nature of the transmission makes
it very wasteful to allocate dedicated transmission resources to all the users.
In contrast to this requirements,
the traditional information theoretic treatment of the multiple-access uplink channel
is focused on few users $K$, large blocklength $n$ and coordinated transmission,
in the sense that each user is given an individual distinct
codebook, and the $K$ users agree on which rate $K$-tuple inside the capacity
region to operate \cite{El1980}.
Mathematically, this is reflected by considering the limit of infinite message- and
blocklength while keeping the rate and the number of users fixed.
Another route, more suited to the IoT requirements,
was taken in recent works like \cite{Che2017,Pol2017}, where the number of users $K$
is taken to infinity along with the blocklength. 
It was shown, that the information theoretic limits
may be drastically different,
when the number of users grows together with the blocklength.\\
A novel random access paradigm, referred to as unsourced random-access (U-RA),
was suggested in \cite{Pol2017}. In U-RA each user employs the same codebook and the task
of the decoder is to recover the list of transmitted messages irrespective of the identity of
the users. The number of \emph{inactive}
users in such a model can be arbitrary large and the performance of the system depends
only on the number of \emph{active} users $K_a$. Furthermore, 
a transmission protocol without the need for
a subscriber identity
is well suited for mass production. These features make U-RA particularly interesting for
the aforementioned IoT applications. \\
In \cite{Pol2017} the U-RA model for the real adder AWGN-MAC was introduced and
a finite-blocklength random coding bound on the achievable energy-per-bit over $N_0$($E_b/N_0$)
was established.
In following works several
practical approaches were suggested which successively reduced the
gap to the random coding achievability bound \cite{Vem2017,Ama2018,Mar2019,Pra2019}.
The model has been extended to fading \cite{Kow2019a} and MIMO channels \cite{Fen2019d}.
A concatenated coding approach for the U-RA problem on the real adder AWGN was proposed
in \cite{Ama2018}. The idea is to split each transmission up
into $L$ subslots. In each subslot the active users send a column from a common
\emph{inner} coding matrix, while the symbols across all subslots are chosen from a
common \emph{outer tree code}. In \cite{Fen2019c} the relation of the inner code to
sparse regression codes (SPARCs) was pointed out.
SPARCs were introduced in \cite{Jos2012} as a class of channel codes for the point-to-point
AWGN channel, which can achieve rates up to Shannon capacity under maximum-likelihood decoding.
Later, it was shown that SPARCs can achieve capacity under
approximate message passing (AMP) decoding with either power
allocation \cite{Rus2017} or spatial coupling \cite{Bar2017a}.
AMP is an iterative low-complexity algorithm for solving random linear
estimation problems or generalized versions thereof \cite{Don2009a,Ran2011,Bay2011}.
A recent survey on SPARCs can be found in \cite{Ven2019a}.\\
Based on the connection of the inner code of \cite{Ama2018} to SPARCs, in \cite{Fen2019c} we
suggested a modified version of AMP as an inner decoder, which improved the performance compared
to the original inner decoder of \cite{Ama2018}.
One of the appealing features of the AMP algorithm is, that it is possible to analyse its
asymptotic performance, averaged over certain random matrix ensembles,
through the so called state evolution (SE) equations \cite{Bay2011,Ber2017a}.
Interestingly the SE equations can also be obtained as the fixed points of the
replica symmetric (RS) potential, an expression that was first calculated through the non-rigorous
replica method \cite{Tan2002,Guo2005c}. It was shown that
in random linear estimation problems the fixpoints of the RS potential also
characterize the symbols-wise posterior distribution of the input elements and therefore
also the error probability of several optimal estimators like the symbol-by-symbol
maximum-a-posteriori (SBS-MAP) estimator \cite{Guo2009,Ran2012a}.
The difference between the AMP and the SBS-MAP estimate is, that
the SBS-MAP estimate always corresponds to the global minimum of the RS-potential, while the
AMP algorithm gets 'stuck' in local minima. The rate below which a local minimum appears
was called the \emph{algorithmic} or \emph{belief-propagation} threshold in \cite{Guo2009, Krz2012a,Bar2017a}.
It was shown in \cite{Krz2012a,Bar2017} that, despite the existence of local minima in the RS-potential,
the AMP algorithm can
still converge to the global minimum when used with spatially coupled matrices.
Although the RS-potential was derived by (and named after) the non-rigorous replica method,
it was recently proven to hold rigorously \cite{Ree2016,Bar2017}. The proof
of \cite{Bar2017} is more general in the sense that it includes the case
where the unknown vector $\sv$
consists of blocks of size $2^J$ and each block is considered to be drawn iid from some distribution
on $\RR^{2^J}$. Initially, the result of \cite{Bar2017} relied on the conjecture that the
SE equations of the AMP algorithm hold for the case of a block iid distribution.
But \cite{Ber2017a} has shown that the SE equations hold under quite weak assumptions
on $\sv$,
which include the block iid case, and therefore has proven the missing conjecture in \cite{Bar2017}.\\
Building on these results, in \cite{Fen2019c}, we calculated the RS-potential
of the inner decoding problem, which allowed us to calculate
the asymptotic error probabilities of the SBS-MAP and the AMP estimate.
The results were semi-analytical, in the sense that the fixpoints of the RS-potential could
only be evaluated numerically.
In this work we show, that in the limit of $K_a,J\to\infty$ with $J = \alpha \log_2 K_a$
for some $\alpha >1$, the RS-potential
converges to a simple form with a sharp threshold on the achievable sum-rates. \\
We have also shown in \cite{Fen2019c} that the inner decoding creates an effective outer
OR-channel \cite{Cha1981,Coh1971}
under a specific input constraint and we gave upper bound on the achievable
rates on that channel. As pointed out in \cite{Fen2019e}, the outer tree code of \cite{Ama2018}
is able to achieve that bound exactly in the limit of infinite subslots $L$ at a decoding
complexity exponential in $L$ or up to a multiplicative constant with a decoding
complexity linear in $L$.\\
Our main contribution in this work is to show that
the concatenated coding scheme of \cite{Ama2018} consisting of multiuser SPARCs combined
with an outer tree code is reliable,
in the sense that it can achieve a vanishing per-user error probability in the
limit of large blocklength and infinitely many users,
at sum-rates up to the symmetric Shannon capacity
$0.5\log_2(1+K_a\SNR)$. This also shows that an unsourced random access scheme can,
in the considered scaling regime $\alpha >1$, achieve the same symmetric
rates as a non-unsourced scheme. \\
The U-RA problem on the real AWGN adder is formally equivalent to the On-Off random access
scheme defined in \cite{Fle2009c}, and there are several other works, which analyse the sparse
recovery problem, assuming an iid prior on the unknown vector, using either the replica method
like \cite{Guo2009,Ree2012a} or more direct compressed sensing based methods like
\cite{Guo2009,Fle2009c,Ree2012a}.
It is not obvious how the asymptotic result of our Theorem \ref{thm:conc} below
can be obtained directly from replica arguments,
since it requires $J$ to scale proportional to the blocklength $n$,
i.e. the undersampling ratio $2^J/n$ to go to infinity.
Such a behavior is not covered by the available
framework.
We can obtain this result by first calculating the RS-potential in
the limit of large $n,L$ with fixed
$J$ and then take the limit $J\to\infty$.
Also compressed sensing based results like \cite{Fle2009c,Ree2012a} are insufficient,
since they contain unspecified
constants, which are necessary to derive an exact capacity.
\section{System model}
\label{sec:system}
Let $K_a$ denote the number of active users,
$n$ the number of available channel uses and
$B = nR$ the size of a message in bits.
The spectral efficiency is given by $\mu = K_aB/n$.
The channel model used is
\beq
\yv = \sum_{i=1}^{K_\text{tot}} q_i \xv_i + \zv, 
\label{eq:basic_channel}
\eeq
where each $\xv_i \in \mathcal{C} \subset \RR^n$
is taken from a common codebook $\mathcal{C}$ and $q_i\in\{0,1\}$ are binary variables
indicating whether a user is active.
The number of active users is denoted as $K_a = \sum_{i=1}^{K_\text{tot}}q_i$.  
The codewords are assumed to be normalized
$\|\xv_i\|_2^2 = nP$ and the noise vector $\zv$ is Gaussian iid
$z_i\sim\mathcal{N}(0,N_0/2)$,
such that $\SNR = 2P/N_0$ denotes the real per-user $\SNR$.
All the active users pick one of the $2^B$ codewords from $\mathcal{C}$,
based on their message $W_k\in[1:2^B]$.
The decoder of the system produces a list $g(\yv)$ of at most $K_a$ messages.
An error is declared if one of the transmitted messages is missing in the output list $g(\yv)$
and we define the per-user probability of error as:
\beq
P_e = \frac{1}{K_a}\sum_{k=1}^{K_a} \PP(W_k \notin g(\yv)).
\eeq
Note that the error
is independent of the user identities in general and especially independent
of the inactive users.
The performance of the system is measured in terms of the standard quantity
$E_b/N_0 := P/(RN_0) $ and the described coding construction is called \emph{reliable}
if $P_e \to 0$ in the considered limit.
\section{Coding Construction}
\label{sec:coding}
In this work we focus on a special type of codebook,
where each transmitted codeword is created in the following way:
First, the $B$-bit message $W_k$ of user $k$ is mapped to an $LJ$-bit codeword from some common \emph{outer}
codebook. Then each of the $J$-bit sub-sequences is mapped to an index $i_k(l) \in [1:2^J]$
for $l=[1:L]$ and $k=[1:K_a]$. The inner codebook is based on a set of $L$
coding matrices $\Am_l\in\RR^{n\times 2^J}$.
Let $\av^{(l)}_i$ with
$i=[1:2^J]$ denote the columns of $\Am_l$. The inner codeword of user $k$
corresponding to the sequence of indices $i_k(1),...,i_k(L)$ is then created as
\beq
\xv_k = \sum_{l=1}^L \av^{(l)}_{i_k(l)}.
\eeq
The $\Am_l$ are assumed to be scaled such that $\|\av^{(l)}_i\|_2^2 = nP/L$.
The above encoding model can be written in matrix form as
\beq
\yv = \sum_{k=1}^{K_a} \Am \mv_k + \zv = \Am \left(\sum_{k=1}^{K_a} \mv_k\right) + \zv.
\label{eq:inner_channel}
\eeq
where $\Am = (\Am_1|...|\Am_L)$ and $\mv_k \in \RR^{L2^J}$ is a binary vector satisfying
$m_{k,(l-1)2^J + i_k(l)} = 1$ and zero otherwise, for all $l=[1:L]$.
Let $\sv = \sum_{k=1}^{K_a} \mv_k$.
\eqref{eq:inner_channel} can be viewed as concatenation of an inner point-to-point channel
$\sv\to \Am\sv + \zv$ and an
outer binary input adder MAC $(\mv_1,...,\mv_{K_a}) \to \sv$.
We will refer to those as the
\emph{inner} and \emph{outer channel}, the corresponding encoder and decoder
will be referred to as \emph{inner} and \emph{outer encoder/decoder} and
the aggregated system of inner and outer encoder/decoder as the \emph{concatenated system}.\\
The per-user inner rate in terms of bits/c.u. is given by
$R_\text{in} := LJ/n$ and the outer rate is given by $R_\text{out} = B/LJ$.

\section{Main Result}
Our main result states that inner and outer codes exist, such that the concatenated coding construction described above
is reliable at sum-rates up to the symmetric Shannon capacity. 
\begin{theorem}
    \label{thm:conc}
    Let $n,L,J,K_a \to \infty$ and $R,\SNR \to 0$ with fixed $E_b/N_0 = \SNR/(2R)$, $S=K_aR$
    and $J = \alpha \log_2 K_a$ for any $\alpha >1$. In this limit there is a concatenated
    code as described above that can be decoded with $P_e \to 0$ if
    \beq
        S < \frac{1}{2}\log_2 (1 + K_a\SNR)
        \label{eq:conc_thm}
    \eeq
    \hfill$\square$
\end{theorem}
Note that within our asymptotic regime $K_a\SNR = 2SE_b/N_0$ is a constant.
As mentioned in Section \ref{sec:coding}, the inner decoding is equivalent to a structured
sparse recovery problem of finding $\sv$ from the knowledge of $\yv$ and $\Am$, where
\beq
    \yv = \Am\sv + \zv
    \label{eq:inner_channel2}
\eeq
and $\sv\in\RR^{L2^J}$ is generated according to the model
described in Section \ref{sec:coding}, i.e. $\sv = \sum_{k=1}^{K_a} \mv_k$.
We say that $\sv$ is \emph{generated from evenly distributed messages}, if
the outer encoded sequences  $i_k(1),...,i_k(L)$ are
distributed evenly, i.e. $\PP(i_k(s) = j) = 1/2^J$
for all $j=[1:2^J]$, and so $\PP(m_{k,(l-1)2^J + i_k(l)} = 1) = 1/2^J$ for all $l=[1:L]$.
We will show that it is enough to recover the support of $\sv$.
The asymptotic limitations of the problem of support recovery of structured sparse vectors
in the considered scaling regime
are a novel result on their
own, therefore we analyse two types of support estimators.
Let $\rhov$ be the binary vector indicating the support of $\sv$, i.e.
$\rho_i = 1$ if and only if $s_i \neq 0$.
The SBS-MAP estimator of $\rhov$
\beq
\hat{\rho}_i = \argmax_{\rho \in \{0,1\}} \PP(\rho_i = \rho|\yv,\Am)
\label{eq:sbs-map}
\eeq
minimizes the SBS error probability $\PP(\hat{\rho}_i \neq \rho_i)$ but is typically unfeasible
to compute in practice.
The second estimator is
the low-complexity AMP algorithm, which produces an estimate of $\rhov$
by iterating the following equations
\beq
\begin{split}
    \rhov^{t+1} &= \eta_t(\Am^\top\zv^t + \rhov^t) \\
    \zv^{t+1}   &= \yv - \Am\rhov^{t+1} +
    \frac{2^JL}{n}\zv^{t}\langle\eta^\prime_t(\Am^\top\zv^t+\rhov^t)\rangle
\end{split}
\label{eq:amp}
\eeq
where the functions $\eta_t:\RR^{2^JL}\to\RR^{2^JL}$ are defined componentwise
$\eta_t(\xv) = (\eta_{t,1}(x_1),...,\eta_{t,2^JL}(x_{2^JL}))^\top$ and each component is given by
\beq
\eta_{t,k}(x) = \sqrt{\hat{P}}\left(1 + \frac{p_0}{1-p_0}\exp\left(\frac{\hat{P}-2\sqrt{\hat{P}}x}{2\tau^2_t}\right)\right)^{-1}
\eeq
with $\tau^2_t = \|\zv^t\|_2^2/n$, $\hat{P} = n\SNR/L$ and $p_0 = (1-2^{-J})^{K_a}$.
$\langle\xv\rangle = (\sum_{i=1}^N x_i)/N$ denotes the average of a vector, $\eta^\prime_t$
denotes the componentwise derivative of $\eta_t$ and
we choose $\rhov^0 = \mathbf{0}$ as initial value.
Our result on the inner recovery problem is as follows:
\begin{theorem}
    \label{thm:inner}
    Let $\Am \in \RR^{L2^J}$ be a matrix with Gaussian iid entries
    $A_{ij}\sim\mathcal{N}(0,P/L)$ and let $\yv$ and $\sv$ be jointly distributed
    according to the model \eqref{eq:inner_channel2} with $\sv$ being generated from
    evenly distributed messages.
    Furthermore, let $R_\text{in} = LJ/n$.
    In the limit $L,n,K_a,J\to \infty$ with $J = \alpha\log_2 K_a$ for some $\alpha>1$
    and $\SNR,R_\text{in} \to 0$ with fixed ratio
    $\mathcal{E}_\text{in}=\SNR/(2R_\text{in})$ and fixed inner sum-rate
    $S_\text{in} = K_aR_\text{in}$ the following holds:\\
    The SBS-MAP detector recovers the support of $\sv$ reliably if 
    \beq
    S_\text{in}\left(1-\frac{1}{\alpha}\right) < \frac{1}{2}\log_2 (1 + 2S_\text{in}\mathcal{E}_\text{in})
    \label{eq:inner_thm_opt}
    \eeq
    and the AMP decoder recovers the support of $\sv$ reliably if
    \beq
    S_\text{in}< \log_2 e\left(1 - \frac{1}{\alpha}\right)^{-1} - \frac{1}{\mathcal{E}_\text{in}}
    \label{eq:inner_thm_alg}
    \eeq
    \hfill$\square$
\end{theorem}
\begin{remark}
    \normalfont
    In the case $K_a = 1$ no outer code is necessary, so $R_\text{in} = R$ and furthermore
    $S_\text{in} = R$ and $2S_\text{in}\mathcal{E}_\text{in} = \SNR$.
    Hence, if $K_a=1$ is fixed and $J\to\infty$,  which corresponds to
    $\alpha \to \infty$, then \eqref{eq:inner_thm_opt} recovers the statements of
    \cite{Jos2012,Bar2017}, i.e. that SPARCs are reliable at rates up to the Shannon
    capacity $0.5\log_2 (1 + \SNR)$
    under optimal decoding. Also the algorithmic threshold \eqref{eq:inner_thm_alg}
    coincides with the
    result of \cite{Bar2017a}.
    In that sense Theorems \ref{thm:conc} and \ref{thm:inner} are
    an extension of \cite{Bar2017a} and show that SPARCs can achieve the optimal rate 
    limit in the unsourced random access scenario. However, notice that the concept of
    our proof technique is simpler, since we make use of the result in \cite{Fen2019c}, which states
    that not only the sections are described by a decoupled channel model, but in the limit
    $J\to\infty$ also the
    individual components. So all the results of Theorem \ref{thm:inner}
    can be derived from the fixpoints of a simple scalar-to-scalar function.
\end{remark}
\begin{remark}
    \normalfont
    The sparse recovery problem \eqref{eq:inner_channel2}
    is very general and it is possible to
    describe random coding for several different classical multiple-access variants,
    where all the users are assumed to have their own codebook.
    For that, let
    $K_a = 1$ and identify the number of section with the number of users. 
    The matrices $\Am_1,...,\Am_L$ are then the codebooks
    of the individual users:
    \begin{itemize}
        \item Fixed $L$ in the limit $J,n\to\infty$ describes the classical
            AWGN adder MAC from \cite{El1980}, where each user has his own codebook.
        \item $L,J,n\to\infty$, where only a fraction of the sections are non-zero
            describes the many-access channel treated in \cite{Che2017}
        \item $J$ fixed and $L,n\to\infty$ describes specific version of
            the many-access MAC treated in \cite{Zad2019,Pol2017}
    \end{itemize}
    It is interesting, that in the first case Theorem \ref{thm:inner} gives the
    correct result, after letting $\alpha\to\infty$, $K_a = 1$ and $L=K$.
    The case of $J,n\to\infty$ at finite $L$ is not directly covered though by our analysis
    framework. Nonetheless, we believe that an extension of this framework should be able
    to show that
    all of the above cases can be derived from a single scalar RS-potential, but this is left
    for future work.
\end{remark}
\begin{proof}[Proof of Theorem \ref{thm:conc}]
    Theorem \ref{thm:inner} shows that, if condition \eqref{eq:inner_thm_opt} is fulfilled,
    there exists an inner coding matrix $\Am$ such that the power constraint is fullfilled on average 
    and the SBS-MAP estimator \eqref{eq:sbs-map}
    recovers the support $\rhov$ of $\sv$ reliably. Then $\rhov$
    is given as the componentwise OR-combination of the input
    message vectors $\mv_k$:
    \beq
        \rhov = \bigvee_{k=1}^{K_a}\mv_k
    \eeq
    This creates an outer noiseless OR-MAC \cite{Cha1981,Coh1971}.
    Let us assume, that all the message vectors $\mv_i$ are
    independently encoded by the same outer code and that the outer encoded
    symbols are evenly distributed.
    The per-user rate of this outer code
    is limited by
    \beq
    K_aR_\text{out}J < 2^J \mathcal{H}_2((1-2^{-J})^{K_a})
    \label{eq:or_bound1}
    \eeq
    where $\mathcal{H}_2$ denotes the binary entropy function.
    As shown in  \cite{Fen2019e}, in the considered limit $K_a,J\to\infty$,
    inequality \eqref{eq:or_bound1} implies
    \beq
    R_\text{out} < 1 - \alpha^{-1}.
    \label{eq:or_bound2}
    \eeq
    Although \eqref{eq:or_bound2} is formally an upper bound it is shown in \cite{Fen2019e}
    that the bound is tight, since it is achievable by the outer tree code described in
    \cite{Ama2018}. Therefore we can assume that a capacity achieving outer code exists
    if $R_\text{out} < 1 - \alpha^{-1}$. Since the total rate is given by
    $R = R_\text{in}R_\text{out}$, we have that $S = S_\text{in}R_\text{out}$ and
    \eqref{eq:conc_thm}
    follows from Theorem \ref{thm:inner}.
\end{proof}
It remains to proof Theorem \ref{thm:inner}.
For that, we build on our results from \cite{Fen2019c},
which characterize the performance of the SBS-MAP estimator \eqref{eq:sbs-map}
in the limit $L,n \to \infty$ with a fixed ratio $L/n$ and fixed $J$.
Through a series of approximations it is shown in \cite{Fen2019c} that for a 
Gaussian iid $\Am$ the
error statistics of the SBS-MAP estimator \eqref{eq:sbs-map}
converge to the error statistics of an SBS-MAP estimate
in $2^JL$ decoupled real Gaussian channels:
\beq
r_i = (\eta\hat{P})^{\frac{1}{2}} s_i + z_i
\label{eq:d_model}
\eeq
where $\hat{P} = n\SNR/L = J\SNR/ R_\text{in} = 2J\mathcal{E}_\text{in}$
and each component $i=[1:L2^J]$
is considered independently of the others. Furthermore,
$s_i \in \{0,1\}$ with 
\beq
p_0 := \PP(s_i = 0) = (1-2^{-J})^{K_a}
\eeq
$\PP(s_i = 1) = 1-p_0$
and
$z_i\sim\mathcal{N}(0,1)$.
The factor $\eta$ is determined by the minimizer of the function
\beq
i^{RS}(\eta) = 2^JI(\eta\hat{P}) + \frac{2^J}{2\beta}[(\eta - 1)\log_2 e - \log_2 \eta],
\label{eq:rs_potential}
\eeq
where $I(\eta\hat{P})$ is the input-output mutual information of the decoupled model
\eqref{eq:d_model} and $\beta = 2^JR_\text{in}/J$. 
The RS potential \eqref{eq:rs_potential} was introduced in \cite{Fen2019c}
as an approximation of the true RS potential of the recovery problem \eqref{eq:inner_channel2},
but it was
shown that the error terms in this approximation are of order $K_a/2^J$. In the asymptotic regime that we consider
$K,J\to\infty$ with $2^J = K_a^\alpha$ and some $\alpha >1$ we have that $K_a/2^J \to 0$.
Therefore, in this limit, \eqref{eq:rs_potential} indeed
characterizes the performance of the SBS-MAP estimator \eqref{eq:sbs-map} exactly.

The AMP algorithm \eqref{eq:amp} is strongly connected to the RS-potential
\eqref{eq:rs_potential} in that the asymptotic
error distribution of the AMP estimate at convergence is described by the same decoupled
channel model \eqref{eq:d_model}, only that 
the coefficient $\eta$ that determines the effective channel strength
is given by the smallest local minimizer of \eqref{eq:rs_potential} \cite{Fen2019c}.
The next Theorem gives the pointwise limit of \eqref{eq:rs_potential}.
\begin{theorem}
    \label{thm:limit}
    In the limit $K_a,J \to \infty$, $R_\text{in},\SNR \to 0$ with fixed ratios
    $\mathcal{E}_\text{in} = \SNR/(2R_\text{in})$, $S=KR_\text{in}$ and $J = \alpha \log_2 K_a$ for 
    some $\alpha > 1$ the pointwise limit of the RS-potential \eqref{eq:rs_potential}
    is given by
    (up to additive or multiplicative terms that are independent of $\eta$ and therefore do not
    influence the critical points of $i^\text{RS}(\eta)$):
    \beq
    \begin{split}
        &i^\text{RS}_\infty(\eta):=\lim_{J\to\infty} i^\text{RS}(\eta) =
        \eta S \mathcal{E}_\text{in}[1-\theta(\eta-\bar{\eta})]\\
        &+\frac{S}{\log_2 e}\left(1-\frac{1}{\alpha}\right)\theta(\eta-\bar{\eta})
        + \frac{1}{2}[(\eta-1)-\ln \eta]
    \end{split}
    \label{eq:rs_limit}
    \eeq
    where
    \beq
        \theta(x):=
        \begin{cases}
            1,\quad \text{if } x > 0\\
            \frac{1}{2},\quad \text{if } x=0\\
            0,\quad \text{if } x<0\\
        \end{cases}
    \label{def:theta}
    \eeq
    and
    \beq
    \bar{\eta} = \frac{1-\frac{1}{\alpha}}{\mathcal{E}_\text{in}\log_2 e}
    \label{def:eta_bar}
    \eeq
    \hfill$\square$
\end{theorem}
\begin{proof}
The RS-potential \eqref{eq:rs_potential}, rescaled by $\beta/2^J$
takes the form
\beq
i^\text{RS}(\eta) = \frac{R_\text{in}2^J}{J}I(\eta\hat{P}) +
\frac{\log_2 e}{2}[(\eta - 1) - \ln \eta]
\label{eq:rs_potential_rescaled}
\eeq
with the mutual information
\beq
I(\eta\hat{P}) := I(X;Y) = H(Y)-H(Y|X)
\eeq
for
$P(X=0) = p_0$, $P(X=1) = 1-p_0$ and $Y = (\eta\hat{P})^\frac{1}{2}X + Z$,
for $Z\sim\mathcal{N}(0,1)$ independent of $X$. 
The mutual information $I(\eta\hat{P})$ can be evaluated as follows.
First, note that in an additive channel $H(Y|X) = H(Z)$,
so $H(Y|X)$ is independent of $\eta$ and therefore we can ignore it.
The distribution of $Y$ is given by
\beq
\begin{split}
    p(y) &= p_0p(y|x=0) + (1-p_0)p(y|x=1)\\ 
         &= \frac{p_0}{\sqrt{2\pi}}\exp\left(-\frac{y^2}{2}\right)\\
         &+ \frac{1-p_0}{\sqrt{2\pi}}\exp\left(-\frac{1}{2}\left(y-(\eta\hat{P})^\frac{1}{2}\right)^2\right),
\end{split}
\eeq
so the differential output entropy $H(Y) = -\int p(y)\log_2 p(y)\mathrm{d}y$ can be split into the
sum of two parts. Define $H_0$ and $H_1$ respectively by
\beq
H_0 := -\frac{1}{\sqrt{2\pi}}\int_{-\infty}^\infty \exp\left(-\frac{y^2}{2}\right)\log_2(p(y))\mathrm{d}y
\label{def:H0}
\eeq
and
\beq
\begin{split}
    H_1 &:= -\frac{1}{\sqrt{2\pi}}\int_{-\infty}^\infty
    \exp\left(-\frac{1}{2}\left(y-(\eta\hat{P})^\frac{1}{2}\right)^2\right)\log_2(p(y))\mathrm{d}y\\
    &= -\frac{1}{\sqrt{2\pi}}\int_{-\infty}^\infty
    \exp\left(-\frac{y^2}{2}\right)\log_2\left(p\left(y+(\eta\hat{P})^\frac{1}{2}\right)\right)\mathrm{d}y
\end{split}
\label{def:H1}
\eeq
such that the following relation holds:
\beq
I(\eta\hat{P}) = p_0H_0 + (1-p_0)H_1.
\eeq
Taking into account the scaling factor in \eqref{eq:rs_potential_rescaled}
and using that $\lim_{J\to\infty}2^J(1-p_0) = K_a$ and $\lim_{J\to\infty}p_0 = 1$
we get that
\beq
\begin{split}
    \lim_{J\to\infty}\frac{R_\text{in}2^J}{J}I(\eta\hat{P}) 
    &= \lim_{J\to\infty}\left(\frac{R_\text{in}2^J}{J}H_0 + \frac{S}{J}H_1\right)
    \label{eq:lim_I}
\end{split}
\eeq
Now let us take a closer look at $\log_2 p(y) = \log_2(e) \ln p(y)$ which appears
in both $H_0$ and $H_1$.
Let $x_1,x_2 > 0$ with $x_2>x_1$. Then for 
the logarithm of the sum of exponentials it holds that
\beq
-\ln (e^{-x_1} + e^{-x_2}) = x_1 + \ln(1+e^{-(x_2-x_1)}).
\eeq
The error term $\ln(1+e^{-(x_2-x_1)})$ decays exponentially as the difference $x_2-x_1$ grows.
Since $p(y)$ is the sum of two exponentials we can approximate $\ln p(y)$ by:
\beq
\begin{split}
    &-\ln p(y) =\\ 
    &\min\left\{\frac{y^2}{2} - \ln(p_0),\frac{1}{2}\left(y-(\eta\hat{P})^\frac{1}{2}\right)^2-\ln(1-p_0)\right\}
\end{split}
\label{eq:max_logsum}
\eeq
This approximation is justified, since the difference of the two exponents
in $p(y)$ is proportional to $\sqrt{J}$, and so it grows large with $J$. 
\footnote{Technically, this approximation does not hold at the point where the two exponents in
    $p(y)$ are equal. However, since the integral of a function does not depend
on the value of the function at points of measure zero, we can redefine $\ln p(y)$ arbitrary 
at that point.}

First, note, that since $\min\{a,b\} \leq a$ and $\min\{a,b\} \leq b$ holds for all $a,b\in \RR$,
$-\ln p(y) \leq y^2/2-\ln(1-p_0)$ as well as $-\ln p(y+(\eta\hat{P})^\frac{1}{2})\leq y^2/2 + \ln (2^J/K)$.
This means that each of the integrands in
$H_0$ and $H_1/J$ resp. is bounded uniformly, for all $J$, by an integrable function.
This allows us to evaluate the integrals by using Lebesgue's theorem on dominated convergence.
For this purpose we need to calculate the pointwise limits of $\ln p(y)$ and
$\ln p(y+(\eta\hat{P})^\frac{1}{2})/J$. The theorem on dominated convergence then states,
that the limit of the integrals is given by the integral of the pointwise limits.\\
The minimum in \eqref{eq:max_logsum} can be expressed as
\beq
-\ln p(y) =
\begin{cases}
    \frac{y^2}{2} &\quad y<\gamma \\
    \frac{1}{2}\left(y-(\eta\hat{P})^\frac{1}{2}\right)^2 +\ln\left(\frac{2^J}{K}\right)&\quad y\geq\gamma
\end{cases}
\eeq
where we neglected $\ln(p_0)=\ln(1 - K/2^J) \sim K/2^J$ and $\gamma$ is given by
\beq
\gamma = \frac{1}{2}\left(\eta\hat{P}\right)^\frac{1}{2} + \ln\left(\frac{2^J}{K}\right)\left(\eta\hat{P}\right)^{-\frac{1}{2}}.
\eeq
Given the considered scaling constraints and
$\hat{P} = J\SNR/R_\text{in} = 2J\mathcal{E}_\text{in}$,
$\gamma$ can be rewritten as
\beq
\gamma = \sqrt{\frac{J}{2}}\left(\sqrt{\eta \mathcal{E}_\text{in}} + \frac{1-\frac{1}{\alpha}}{\log e\sqrt{\eta \mathcal{E}_\text{in}}}\right)
\eeq
The term in parenthesis is strictly positive for all $\eta$ so $\lim_{J\to\infty}\gamma = \infty$
and therefore the pointwise limit of $\ln p(y)$ is give by $\lim_{J\to\infty} \ln p(y) = -y^2/2$.
It follows from
Lebesgue's theorem on dominated convergence that
\beq
\lim_{J\to\infty} H_0 = \log_2 e
\label{eq:lim_H0}
\eeq
which is independent of $\eta$, so we can ignore it when evaluating $i^\text{RS}(\eta)$.
For the calculation of $H_1$ note that:
\beq
-\ln p\left(y+(\eta\hat{P})^\frac{1}{2}\right) =
\begin{cases}
    \frac{1}{2}\left(y+(\eta\hat{P})^\frac{1}{2}\right)^2 &\quad y<\gamma'\\
    \frac{y^2}{2} + \ln \left(\frac{2^J}{K}\right) &\quad y\geq\gamma'
\end{cases}
\eeq
where we defined
$\gamma' := \gamma - (\eta\hat{P})^\frac{1}{2}$. $\gamma'$ is not non-negative anymore and
therefore the asymptotic behavior of $\gamma'$ depends
on $\eta$ in the following way:
\beq
\lim_{J\to\infty}\gamma' = 
\begin{cases}
    \infty  &\ \text{if } \eta < \bar{\eta}\\
    0       &\ \text{if } \eta = \bar{\eta}\\
    -\infty &\ \text{if } \eta > \bar{\eta}\\
\end{cases}
\eeq
where $\bar{\eta}$ was defined in \eqref{def:eta_bar}.
This gives the following asymptotic behavior:
\beq
-\lim_{J\to\infty}\frac{\ln p(y + (\eta\hat{P})^\frac{1}{2})}{J} = 
\begin{cases}
    \eta \mathcal{E}_\text{in} &\ \eta < \bar{\eta}\\
    (1-\alpha^{-1})/\log_2 e &\ \eta \geq \bar{\eta}
\end{cases}
\label{eq:lim_ln}
\eeq
Finally, using \eqref{eq:lim_H0}, \eqref{def:H1}, \eqref{eq:lim_I}, \eqref{eq:lim_ln} and the $\theta$ function defined in
\eqref{def:theta} we get:
\beq
\begin{split}
    &\lim_{J\to\infty}\left(\frac{i^\text{RS}(\eta)}{\log_2 e} - \frac{R_\text{in}2^J}{J}\right) 
    = \eta S\mathcal{E}_\text{in}[1-\theta(\eta-\hat{\eta})] \\
    &+ S \left(1-\frac{1}{\alpha}\right)\theta(\eta-\bar{\eta})
+\frac{1}{2}\left[(\eta - 1) - \ln \eta\right]
\end{split}
\eeq
This proofs the statement of the theorem. 
\eqref{eq:max_logsum}.
\end{proof}
With Theorem \ref{thm:limit} we can proof Theorem \ref{thm:inner}
and conclude the proof of Theorem \ref{thm:conc}.
\begin{proof}[Proof of Theorem \ref{thm:inner}]
    We have discussed that the error probability of the SBS-MAP detector is specified
    by $\eta^*\hat{P} = \eta^*2\mathcal{E}_\text{in}J$,
    the effective channel strength in the decoupled model \eqref{eq:d_model},
    where $\eta^*$ is the 
    global minimizer of $i^\text{RS}(\eta)$ in the interval $[0,1]$. In a similar fashion
    the error probability of the AMP decoder \eqref{eq:amp} at convergence is described by
    $\eta^*_\text{loc}\hat{P}$, where $\eta^*_\text{loc}$ is the smallest local minimizer
    of $i^\text{RS}(\eta)$. \\ 
    By Theorem \ref{thm:limit} the derivative of $i^\text{RS}_\infty(\eta)$ in
    \eqref{eq:rs_limit} is given by
    \beq
    \frac{\partial i^\text{RS}_\infty}{\partial \eta}(\eta) = S\mathcal{E}_\text{in}[1-\theta(\eta-\bar{\eta})] + \frac{1}{2}\left(1-\frac{1}{\eta}\right)
    \eeq
    for $\eta\neq\bar{\eta}$.
    The critical points of the derivative are
    \beq
        \eta_0^* = (1+2S\mathcal{E}_\text{in})^{-1}
    \eeq
    and
    \beq
        \eta_1^* = 1.
    \eeq
    Note that the first point $\eta_0^*$ is critical if and only if $\eta_0^* < \bar{\eta}$,
    which,
    after rearranging, gives precisely condition \eqref{eq:inner_thm_alg}.
    Also note, that the second derivative of
    $i^\text{RS}_\infty$ is $(4\eta)^{-2}$, so it is non-negative for all $\eta>0$.
    Therefore the critical points are indeed minima. A local maximum may appear only at $\eta=\bar{\eta}$
    where $i^\text{RS}_\infty$ is not differentiable.
    The values of $i^\text{RS}_\infty$ at the minimal points are
    \beq
    \begin{split}
    i^\text{RS}_\infty(\eta^*_0) 
    &= \frac{S\mathcal{E}_\text{in}}{1 + 2S\mathcal{E}_\text{in}} + 
    \frac{1}{2}\left[\frac{-2S\mathcal{E}_\text{in}}{1 + 2S\mathcal{E}_\text{in}}+\ln(1+2S\mathcal{E}_\text{in})\right] \\
    &= \frac{\log_2 (1 + 2S\mathcal{E}_\text{in})}{2\log_2 e}
    \end{split}
    \eeq
    if $\eta^*_0< \bar{\eta}$, and
    \beq
    i^\text{RS}_\infty(\eta^*_1) = \frac{S}{\log_2 e}\left(1 - \frac{1}{\alpha}\right)
    \eeq
    It is apparent that $i^\text{RS}_\infty(\eta^*_1)$ is the global minimum if and only if condition
    \eqref{eq:inner_thm_opt} is fulfilled.
    We implicitly used here that $\bar{\eta}\leq 1$,
    that is because condition \eqref{eq:inner_thm_opt} implies $\bar{\eta}< 1$,
    which can be seen by solving inequality \eqref{eq:inner_thm_opt} for $\mathcal{E}_\text{in}$.
    If $\eta^*_1 = 1$ is indeed the global minimizer of \eqref{eq:rs_limit}, 
    the effective power in the decoupled channel \eqref{eq:d_model} is given by $\hat{P}$. 
    Since $\hat{P}$ grows proportional to $J$,
    the effective power in the channel and therefore also the probability of misestimating
    the support go to zero with $J\to\infty$. This concludes the proof of Theorem \ref{thm:inner}.
\end{proof}
\section{Conclusion}
We have shown that the concatenated coding construction in \cite{Ama2018,Fen2019c}
is asymptotically
optimal as the blocklength $n$, the number of active users $K_a$,
the number $L$ and the size $J$ of the subslots go to infinity. This makes the
SPARC based coding construction the first of the known U-RA codes to have an 
asymptotic optimality guarantee. Our result also shows more generally that 
the achievable trade-off between sumrate and $E_b/N_0$ in U-RA converges to the 
Shannon bound \eqref{eq:conc_thm} in the considered limit.
\printbibliography
\balance
\end{document}